\documentclass[11pt]{amsart}
\usepackage{latexsym,amssymb,amsmath}
\usepackage{graphicx}
\usepackage{epsfig}
\usepackage{amsthm}
\usepackage{amsmath}
\usepackage{amsfonts}
\usepackage{amssymb}
\usepackage{amscd}
\usepackage{mathrsfs}
\usepackage[all]{xypic}
\usepackage{url}

\newtheoremstyle{custom}
  {3pt}
  {3pt}
  {\slshape}
  {}
  {\bfseries}
  {.}
  { }
   {}
\theoremstyle{custom}
\newtheorem{theorem}{Theorem}[subsection]

\newtheorem{proposition}[theorem]{Proposition}
\newtheorem{proposition/definition}[theorem]{Proposition/Definition}
\newtheorem{lemma}[theorem]{Lemma}

\theoremstyle{definition}
\newtheorem{definition}[theorem]{Definition}

\newtheorem{example}[theorem]{Example}

\theoremstyle{remark}

\newtheoremstyle{exercise}
  {3pt}
  {6pt}
  {}
  {}
  {\bfseries}
  {:}
  { }
   {}
\theoremstyle{exercise}
\newtheorem{exercise}[theorem]{Exercise}
\newtheoremstyle{exercises}
  {3pt}
  {6pt}
  {}
  {}
  {\bfseries}
  {:}
  {\newline}
   {}

\theoremstyle{exercise}
\newtheorem{exercises}[theorem]{Exercises}

\input epsf
\def\boxit#1{\vbox{\hrule height1pt\hbox{\vrule width1pt\kern3pt
  \vbox{\kern3pt#1\kern3pt}\kern3pt\vrule width1pt}\hrule height1pt}}

\def\BC{\mathbb C}\def\BS{\mathbb S}

\def\hd{,...,}

\def\cS{{\mathcal S}}

\def\11{\mathbf 1}

\def\a{\alpha}

\def\b{\beta}

\def\ot{{\mathord{ \otimes } }}

\def\otc{{\mathord{\otimes\cdots\otimes}\;}}

\def\ra{{\mathord{\;\rightarrow\;}}}

\def\La#1{\Lambda^{#1}}

\def\ep{\epsilon}

\def\t{\tau}

\def\a{\alpha}
\def\b{\beta}

\def\FS{\mathfrak  S}

\def\BC{\mathbb  C}

\def\BS{\mathbb  S}

\def\ep{\epsilon}

\def\hd{, \hdots ,}

\def\La#1{\Lambda^{#1}}

\def\ra{\rightarrow}

\def\tpfaff{\operatorname{Pf}}

\def\be{\begin{equation}}
\def\ene{\end{equation}}

\def\sgn{{\rm{sgn}}}

\newcommand{\isom}{\cong}

\def\dual{{^\vee}}

\newcommand{\Id}{\operatorname{Id}}

\newcommand{\sPf}{\operatorname{sPf}}
\newcommand{\tprod}{\mathop{\otimes}}

\def\tpfaff{{\rm Pfaff}}

\def\tcr{{\rm cr}}
\def\sgn{{\rm sgn}}
\def\G{\Gamma}

\begin{document}
\title{Holographic algorithms without matchgates}
\author{J.M.~Landsberg,  Jason Morton and Serguei Norine}
\date{\today}
\maketitle
\begin{abstract}The theory of
{\it holographic algorithms},  which are polynomial time algorithms for certain combinatorial counting problems, yields insight into the hierarchy of complexity classes.  In particular, the theory  produces    algebraic tests for a problem to be in the class $\bold P$.  In this article we streamline the implementation of holographic algorithms by eliminating one of the steps in the construction procedure, and generalize their applicability to new signatures. Instead of {\it matchgates}, which are weighted graph fragments that replace vertices of a natural bipartite graph $\Gamma_P$ associated to a problem $P$, our approach uses only a natural number-of-edges by number-of-edges matrix associated to $\G_P$.  An easy-to-compute multiple of its Pfaffian is the number of solutions to the counting problem.  This simplification improves our understanding of the applicability of  holographic algorithms, indicates a more geometric approach to complexity classes, and facilitates practical implementations. The generalized applicability arises because our approach allows for new algebraic tests that are different from the ``Grassmann-Pl\"ucker identities'' used up until now.  Natural problems treatable by these new methods have been previously considered in a different context, and we present one such example.
\end{abstract}

\section{Introduction}
In \cite{MR2120307,MR1932906,ValiantSimulatingQCiPT,ValiantFOCS2004,MR2184617,MR2386281} L. Valiant introduced
{\it matchgates} and  {\it holographic algorithms},
in order to prove the existence of polynomial time
algorithms for  counting and sum-of-products problems that na\"\i vely appear to have  exponential complexity.
Such algorithms have been studied in depth and further developed  by J. Cai et. al. \cite{MR2305569, MR2277247,MR2354219,MR2402465,MR2424719,MR2362482,MR2417594,CaiLX08FOCS}.

The algorithms work as follows: suppose the problem $P$ is
to count the number of
satisfying assignments to a collection of Boolean variables $x_1\hd x_m$ subject
to clauses $c_1\hd c_p$. The problem $P$ defines a bipartite graph $\Gamma_P=(V,U,E)$,
with vertex sets $V=\{x_1\hd x_m\}=\{x_i\}$ and $U=\{c_1\hd c_p\}=\{c_s\}$ and there is an edge $(i,s)\in E$ iff $x_i$ appears in $c_s$.
Holographic algorithms apply if the coordinates of the clauses and variables, expressed as tensors, satisfy a collection of polynomial equations called {\it matchgate identities}, possibly, in fact usually, after a change of basis.
In the matchgate approach   each vertex of $\Gamma_P$ is replaced by a weighted graph fragment called a {\it matchgate} to form a new
weighted graph $\G_{\Omega(P)}$, such that the weighted sum of perfect matchings of $\G_{\Omega(P)}$ equals the number of satisfying assignments to $P$. If $\G_{\Omega(P)}$ is planar, or more generally {\it Pfaffian}, the weighted sum of perfect matchings  of $\G_{\Omega(P)}$ can be computed in time polynomial in $|E|$ using the FKT algorithm
\cite{MR0253689,MR0136398}.  FKT defines a sign-altered skew-symmetric adjacency matrix $X$ of $\G_{\Omega(P)}$ whose Pfaffian equals the  weighted sum of matchings of $\G_{\Omega(P)}$. Our approach
is related directly to the graph $\G_P$, computing the Pfaffian of a natural $|E| \times |E|$ 
matrix associated to $\Gamma_P$.  We replace FKT with an edge ordering defined by a plane curve as described in Section \ref{signfixsect}.  Evaluating the Pfaffian takes polynomial time.

Equivalently, the number of satisfying assignments to $P$ is the result of the pairing of a vector $G\in \BC^{2^{|E|}}$  formed as the tensor product of \lq \lq local\rq\rq\  data representing the variables and a vector $R$ in the dual vector space, the tensor product of  \lq \lq local\rq\rq\   data concerning the clauses (see Section \ref{sect41}).
The Valiant-Cai formulation of holographic algorithms can be summarized as
\begin{equation} \label{pfholant}
\# \text{satisfying assignments of $P$} = \langle G, R\rangle =    \text{weighted sum of perfect matchings of  $\Gamma_{\Omega(P)}$}.
\end{equation}

In this article we give a   new construction which   eliminates the need to construct matchgates.
We associate constants $\a=\a_G,\b=\b_R$ (depending only on the number of each type of vertex) and $|E|\times |E|$-skew symmetric matrices $\tilde z=\tilde z_G,y=y_R$ directly to
$G$, $R$, {\it without the construction of matchgates}, to obtain the equality:
\begin{equation} \label{ourholant}
\# \text{satisfying assignments of $P$} = \langle G, R\rangle = \alpha \beta \tpfaff(\tilde z + y);
\end{equation}
see Examples \ref{ex:first} and \ref{ex:second}.
The constants and matrices are essentially just components of the vectors $G,R$.
The algorithm complexity is dominated by evaluating the Pfaffian.

The key to our approach is that
a vector satisfies the matchgate identities iff it is a vector of sub-Pfaffians of some skew-symmetric matrix,
and that the pairing of two such vectors can be reduced to calculating a Pfaffian of a new matrix
constructed from the original two. A similar phenomenon holds in great generality discussed in
Appendix \S\ref{spinsect}. A simple example is the set of vectors of sub-minors of an arbitrary rectangular
matrix. We describe an example of such an implementation in Section \ref{CLexample}.

The starting point of our investigations was the observation that the matchgate identities
come from classical geometric objects called {\it spinors}. The results in this
article do not require any reference to spinors to either state or prove, and for the convenience of
the reader not familiar with them we have eliminated all mention of them
except for this paragraph and an Appendix (\S\ref{spinsect}), included for the interested reader.
However further results, such as our characterization of $1$-realizable signatures \cite{LMsig},
do require use of the representation theory of the spin groups.

\smallskip

\section{Counting problems as tensor contractions} \label{sec:counting}
For brevity we continue to restrict to problems $P$ counting the number of satisfying assignments of Boolean variables $x_i$ subject to clauses $c_s$ (such as \#Pl-Mon-NAE-SAT in Example \ref{ex:NAESATPart1}).
Following e.g., \cite{MR2277247} express $P$ in terms of a tensor contraction diagrammed by a planar bipartite graph $\Gamma_P=(V,U,E)$ as
above (see  Figure \ref{fig:bipartite}), together with the data of   tensors $G_i=G_{x_i}$ and $R_s=R_{c_s}$   attached at each vertex $x_i \in V$ and $c_s \in U$. $G_i$ will record that $x_i$ is $0$ or $1$ and
$R_s$ will record that the clause $c_s$ is satisfied.   Let
$
n = |E|
$
be the number of edges in $\Gamma_P$.

For each edge $e=(i,s) \in E$ define a $2$-dimensional vector space $A_{e}$ with
basis $a_{e|0},a_{e|1}$. Say $x_i$ has degree $d_i$ and is joined to $c_{j_1}\hd c_{j_{d_i}}$.
Let $E_i$ denote the set of edges incident to $x_i$ and associate to each $x_i$ the tensor
\begin{align}\label{xiconsistenteeqn}
G_i:=&a_{i,s_{j_1}|0}\otc a_{i,s_{j_{d_i}}|0}+ a_{i,s_{j_1}|1}\otc a_{i,s_{j_{d_i}}|1}\in A_i:=
  A_{is_{j_1}}\otc A_{is_{j_{d_i}}}\\
 =& \tprod_{e \in E_i} a_{e|0} +  \tprod_{e \in E_i} a_{e|1} \notag
\end{align}
   The  tensor $G_i$ represents that either $x_i$ is true (all $1$'s) or false (all $0$'s).  It is called a {\it generator} and
in the matchgates literature
is  denoted  by the vector $(1,0\hd 0,1)$ corresponding to a lexicographic basis of $A_i=\tprod_{e \in E_i} A_e$.
This vector is called
  its {\it signature}.  We use notation emphasizing the tensor product structure of the vector space $A_i=\BC^{2^{d_i}}$,
and will use the word signature to refer to the tensor expression of $G_i$.

Next define a tensor associated to each clause $c_s$ representing that $c_s$ is
satisfied.
Let $A_{e}^*$ be the dual space to $A_e$ with dual basis $\a_{e|0},\a_{e|1}$. Let $E_s$ denote
the set of edges incident to $c_s$.
For example, if $c_s$ has degree $d_s$ and is \lq\lq not all equal\rq\rq\ (NAE), then the corresponding
tensor (called a {\it recognizer}) associated to it is
\begin{equation}\label{sisnae}
R_s:=\sum_{(\ep_1\hd \ep_{d_s})\neq (0\hd 0),(1\hd 1)}
\a_{i,{s_1}|\ep_1}\otc \a_{i,{s_{d_s}}|\ep_{d_s}}
=
\sum_{(\ep_1\hd \ep_{d_s})\neq (0\hd 0),(1\hd 1)}
\tprod_{e \in E_s} \a_{e|\epsilon_e}
\end{equation}

Now consider $G\! := \! \ot_iG_i$ and $R\!:=\! \ot_sR_s$ respectively elements of the vector spaces $A\!:=\! \ot_{e}A_{e}$ and $A^* \! :=\! \ot_{e}A^*_{e}$.
Then the number of satisfying assignments to $P$ is $\langle G,R\rangle$
where $\langle\cdot,\cdot\rangle: A\times A^*\ra \BC$ is   the pairing
of dual vector spaces.
At this point we have merely exchanged our original counting problem for
the computation of a  pairing in vector spaces of dimension $2^{|E|}$.
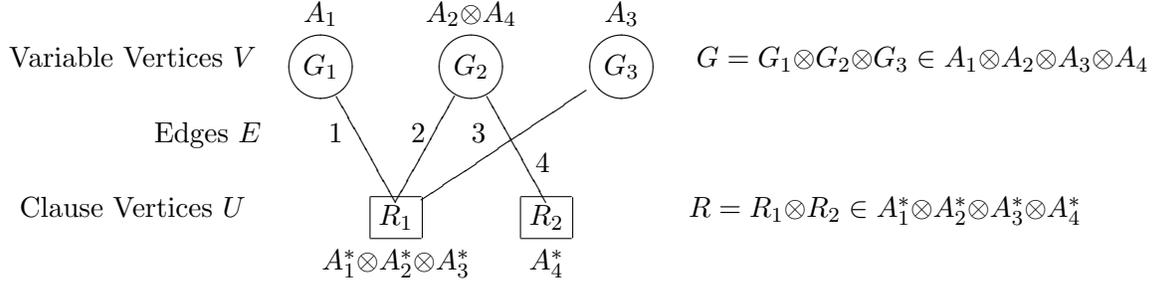
\begin{figure}[htb]
\[
\begin{xy}<10mm,0mm>:
(0,2) *++{G_1}*\frm{o};
  p+(1,-1.8) **@{-},
  p+(0,.6) *+!{A_1},
(2,2) *++{G_2}*\frm{o};
  p+(1,-1.8) **@{-},
  p+(-1,-1.8) **@{-},
  p+(0,.6) *+!{A_2 \ot A_4},
(4,2) *++{G_3}*\frm{o};
  p+(0,.6) *+!{A_3},
(1,0) *+{R_1}*\frm{-} **@{-};
  p+(0,-.7) *+!{A_1^* \ot A_2^* \ot A_3^*},
(3,0) *+{R_2}*\frm{-};
  p+(0,-.7) *+!{A_4^*},
(-2.5,2) *+!{\text{Variable Vertices } V};
(-1.5,1) *+!{\text{Edges } E};
(-2.5,0) *+!{\text{Clause Vertices } U};
(8,2) *+!{G=G_1 \ot G_2 \ot G_3 \in A_1 \ot A_2 \ot A_3 \ot A_4};
(7.5,0) *+!{R=R_1 \ot R_2 \in A_1^* \ot A_2^* \ot A_3^* \ot A_4^*};
(0.2,1) *+!{1};
(1.3,1) *+!{2};
(2.1,1) *+!{3};
(2.95,.6) *+!{4};
\end{xy}
\]
\caption{A bipartite graph $\Gamma$ diagrams a tensor contraction, representing an exponential sum of products such as counting the satisfying assignments of a satisfiability problem.  Boxes denote clauses, circles denote variables.  Each clause or variable corresponds to a tensor lying in the indicated vector space; e.g.\ $R_1 \in A_1^* \ot A_2^* \ot A_3^*$.
Instead of replacing each vertex with a matchgate, our construction defines an $n \times n$ matrix, where $n$ is the number of edges in the problem graph.  The Pfaffian of this matrix, times a constant depending on the number of each type of variable and clause, is the number of satisfying assignments.
 } \label{fig:bipartite}
\end{figure}
\section{Local conditions and change of basis} \label{sec:GPandbasis}
In order to be able to construct the matchgates corresponding to the $x_i,c_s$, there are local conditions that need to be satisfied; the algebraic equations placed on the $G_i,R_s$ are called the  {\it Grassmann-Pl\"ucker identities} (or {\it Matchgate Identities} in this context).
See, e.g., Theorem 7.2 of \cite{MR2402465} for an explicit expression of the equations, which are originally due to Chevalley in the 1950's \cite{MR1636473}.
These identities ensure that a tensor $T$ representing a variable or clause can be written as a vector of sub-Pfaffians of some matrix.
From the matchgates point of view, these equations are necessary and sufficient conditions for the existence of graph fragments
that can replace the vertices of $\G_P$ to form a new weighted graph $\G_{\Omega(P)}$ such that
the weighted perfect matching polynomial of $\G_{\Omega(P)}$ equals $\langle G,R\rangle$.

Expressed in the basis most natural for a problem, a clause or variable tensor   may fail to satisfy the Grassmann-Pl\"ucker identities.  However it may do so under a change of basis; e.g.\ in Example \ref{ex:NAESATPart1}, we replace the basis (True, False) with (True$+$False, False$ - $True). Such a change of basis will not change the value of  the pairing $A\times A^*\ra \BC$   as long as we make the corresponding dual change of basis in the dual vector space---but of course this may cause the tensors in the dual space to fail to satisfy the identities. Thus one needs a change of basis that works for both generators and recognizers.  In this article, as in almost all existing applications of the theory, we only consider changes of bases in the individual $A_e$'s, and we will perform the exact same change of basis in each such, although neither restriction is {\it a priori}  necessary for the theory.

\subsection{Example}\label{ex:NAESATPart1}
In \#Mon-3-NAE-SAT, we are given a Boolean formula in conjunctive normal form where each clause has exactly three literals, and all are either positive or negative (no mixed negations). A clause is satisfied if it contains at least one true and one false literal.    The counting problem asks how many satisfying truth assignments to the variables exist.  The generator tensor $G_i$ corresponding to a variable vertex $x_i$ is
\eqref{xiconsistenteeqn}. The recognizer tensor  corresponding to a NAE clause $R_s$ is
\eqref{sisnae} and in our case we will have $d_s=3$ for all $s$.

Let $T_0$ be the basis change, the same in each $A_e$, sending $a_{e0} \mapsto a_{e|0}+a_{e|1}$ and $a_{e|1} \mapsto a_{e|0}-a_{e|1}$
which induces the basis change $\a_{e|0} \mapsto \frac{1}{2}(\a_{e|0}+\a_{e|1})$ and $\a_{e|1} \mapsto \frac{1}{2}(\a_{e|0}-\a_{e|1})$ in $A^*_e$.  This basis is denoted $\mathbf{b2}$ in \cite{MR2386281}. Applying $T_0$, we obtain
\[
T_0
(a_{i,s_{i_1}|0}\otc a_{i,s_{i_{d_i}}|0}+ a_{i,s_{i_1}|1}\otc a_{i,s_{i_{d_i}}|1})
=2\sum_{\{(\ep_1\hd \ep_{d_i}) \mid \sum \ep_\ell=0 \,(\text{mod } 2) \}}
a_{i,s_{i_1}|\ep_1}\otc a_{i,s_{i_{d_i}}|\ep_{d_i}}.
\]
In the matchgates literature this tensor is denoted by the vector $(2,0,2,0\hd 2,0,2)$ (assuming the number of incident edges is even).
   We also have
\begin{align*}
&T_0 \left ( \sum_{(\ep_1,\ep_2,\ep_3)\neq (0,0, 0),(1,1, 1)}
\a_{i,{s_1}|\ep_1}\ot  \a_{i,{s_{2}}|\ep_{2}} \ot \a_{i,{s_{3}}|\ep_{3}}
\right )\\
&=
6\a_{i,{s_1}|0}\ot  \a_{i,{s_{2}}|0} \ot \a_{i,{s_{3}}|0}
-2(\a_{i,{s_1}|0}\ot  \a_{i,{s_{2}}|1} \ot \a_{i,{s_{3}}|1} +
\a_{i,{s_1}|1}\ot  \a_{i,{s_{2}}|0} \ot \a_{i,{s_{3}}|1} +
\a_{i,{s_1}|1}\ot  \a_{i,{s_{2}}|1} \ot \a_{i,{s_{3}}|0}
)
\end{align*}
or denoted by its coefficients, $(6,0,0,0,-2,-2,-2,0)$.

\section{Holographic algorithms without matchgates}

Though we do not use matchgates, in our approach the matchgate identities still must be satisfied under a change of basis as above.
Our purpose is to make $G,R$ expressible as vectors of sub-Pfaffians of
some skew-symmetric matrices. While this appears to be a global condition, it can
be accomplished locally.  Let $u,v$ be $s\times s$ and $t\times t$ matrices,
and form a block diagonal $(s+t)\times (s+t)$ matrix from them.  The vector of sub-Pfaffians of the new block diagonal matrix in $\BC^{2^{s+t}}$ can be   obtained by taking the $2^s\times 2^t$ matrix corresponding to the   product
of the (column) vector  of sub-Pfaffians of $u$ with the (row) vector of sub-Pfaffians of $v$, and writing the matrix as a vector in $\BC^{2^{s+t}}$.
The analogous statement holds for block diagonal matrices built out of an arbitrary number of smaller matrices.
Thus if each $G_i,R_s$ is a vector of Pfaffians,    the corresponding $G,R$ will be so too; see Proposition \ref{prop:locglob}.
 Theorem \ref{thm:tildesumPfaff} below shows how realizing $G,R$ as vectors of sub-Pfaffians aids one in computing the pairing
$\langle G,R\rangle$ indicated by \eqref{ourholant}.

In all this there is a problem of signs that we have not yet discussed. The problem arises because
if we order the $x_i$ and $c_s$, there are two natural types of orders of the vector spaces in the
tensor products of the $A_{is}$, one grouping   $i$'s and one grouping   $s$'s. The block-diagonal
discussion above cannot be simultaneously applied to both orderings at once. We explain this problem
in detail and how to overcome it in \S\ref{signfixsect}.

\subsection{The complement pairing and representing $G$ and $R$ as vectors of sub-Pfaffians}\label{sect41}
Assume we have a problem expressed as above and have
constructed tensors $G,R$ such that in some change of basis their component tensors $G_i,R_s$ satisfy
the Grassmann-Pl\"ucker identities.
For the purposes of exposition, we will assume the total number of
edges is even. See the discussion in the Appendix \S\ref{spinsect} for the case of an odd number
of edges.

To compute $\langle G,R\rangle$, we will represent $G$ and $R$ as vectors of sub-Pfaffians.  For an $n \times n$ skew-symmetric matrix $z$, the vector of sub-Pfaffians $\sPf(z)$ lies in a vector space  of dimension $\BC^{2^n}$, where the coordinates are labeled by subsets $I \subset [n]$, and
\[
(\sPf(z))_I = \tpfaff(z_I)
\]
where $z_I$ is the submatrix of $z$ including only the rows and columns in the set $I$.  Letting $ I^C =  [n] \setminus I$, similarly define $\sPf^{\dual} \in \BC^{2^n}$ by
\[
(\sPf{\dual}(z))_I = \tpfaff(z_{I^C}).
\]

The vector spaces $A,A^*$ come equipped with un-ordered bases induced from the bases of the $A_e$. These bases do not have a canonical
identification with   subsets of $(1\hd n)$ but do have a convenient choice of
identification after making a choice of edge ordering.  After an ordering $\bar{E}$ of the edges has been chosen we obtain
ordered bases of $A,A^*$.
To obtain the convenient choice of identification  for $A$, identify the vector corresponding to $I=(i_1\hd i_{2p})$ with
the element with $1$'s in the $i_1\hd i_{2p}$ slots and zeros elsewhere,
so, e.g. $I=\emptyset$ corresponds to $(0\hd 0)$, ..., $I=(1\hd 2n)$ corresponds to $(1\hd 1)$.
Reverse the correspondence for  $A^*$.

For later use, we remark that
with these identifications, as long as the first (resp. last) entry of $G$ (resp. $R$) is
non-zero, we may rescale to normalize them to be one. (If say, e.g., the first entry of $G$ is
zero but the last is not, and last entry of $R$ is non-zero, we can just reverse the identifications
and proceed.) Note that the first and last choices of entries are independent of the edge ordering, but
if necessary, to get the first and last entries non-zero, we simply take a less convenient choice
of identification. (See \S\ref{spinsect} for an explanation of this freedom.)
As long as this is done consistently it will not produce any problems.

{\it
In the rest of this section we assume that the
local problem has been solved, i.e., that  Grassmann-Pl\"ucker identities hold for
all the $G_i,R_s$ possibly after  a change of basis.
}  We also assume for brevity that the $G_i$ and $R_s$ are symmetric, i.e.
that $G_i=\sPf(x_i) = \sPf( \pi (x_i))$ for any permutation $\pi$ on
the edges incident on $G_i$; this covers most problems of interest.
For the more general case when the variables or clauses are not
symmetric, and we need to be more careful about defining $\bar{E}_G$ and $\bar{E}_R$, see Section \ref{signfixsect} and the Appendix \S\ref{nonsym}.

\begin{definition} \label{gorodef}Call an edge order
such that edges incident on each $x_i \in V$ (resp. $c_s\in U$) are adjacent   a {\it generator order}
(resp. {\it recognizer order})
and denote such by $\bar{E}_G$
(resp. $\bar{E}_R$).
\end{definition}

\begin{proposition} \label{prop:locglob}
Suppose $P$ is a counting problem as above,  $\bar{E}_G$ and   $\bar{E}_R$ are
respectively generator and recognizer orders.  If for all $x_i \in V$ there exists
 $z_i\in Mat_{d_i\times d_i}$ such that $\sPf(z_i) = G_i$ under the $\bar{E}_G$ identification, and similarly,
there exists $y_s\in  Mat_{d_s\times d_s}$ for $c_s$ and $R_s$ with $\sPf^{\vee}(y_s) = R_s$, then there exists $z,y\in Mat_{|E|\times |E|}$ such that
\[
\sPf(z) = G \qquad \text{under the $\bar{E}_G$ identification and}
\]
\[
\sPf^{\dual}(y) = R \qquad \text{under the $\bar{E}_R$ identification;}
\]
$z,y$ are just given by stacking the component matrices $z_i,y_s$ block-diagonally.
\end{proposition}

As the proposition suggests, a difficulty appears when we try to find an order $\bar{E}$ that works for both
generators and recognizers.

\begin{definition}\label{validdef}
An order $\bar{E}$ is {\it   valid} if there exists skew-symmetric matrices $z,y$   such that
$\sPf(z) = G$  and $\sPf\dual(y) = R$ under the $\bar{E}$ identification.
\end{definition}

Thus if an order is   valid
$$
\langle G,R\rangle =\sum_I\sPf_I(z)\sPf_{I^c}(y)
$$
and in the next subsection we will see how to evaluate the right hand
side in polynomial time.
Then in \S\ref{signfixsect} we   prove that if $\Gamma_P$ is planar, there
is always a valid ordering.

\subsection{Evaluating the complementary pairing of vectors of sub-Pfaffians}
Let $n$ be even.
For an even set $I \subseteq [n]$, define $\sigma(I)=\sum_{i \in I}i$, and define $\sgn(I)=(-1)^{\sigma(I)+|I|/2}$.
Proofs of the following lemma can be found in \cite[p. 110]{MR1069389} and \cite[p. 141]{MR1713476}.

\begin{lemma} \label{lem:sumPfaff}
Let $z$ and $y$ be skew-symmetric $n \times n$ matrices. Then
\[
\tpfaff(z + y) = \sum_{p=0}^n \sum_{I \subseteq [n],\\|I|=2p} \sgn(I) \tpfaff_{I}(z) \tpfaff_{I^C}(y)
\]
\end{lemma}

To use Lemma \ref{lem:sumPfaff} to compute inner products   we need to adjust one of the matrices to correct the signs. For a matrix $z$
define a matrix $\tilde{z}$ by setting $\tilde{z}^i_j = (-1)^{i+j+1}z^i_j$.
Let $z$ be an $n \times n$ skew-symmetric matrix. Then for every even $I \subseteq [n]$,
\[
\tpfaff_I(\tilde{z}) = \sgn(I) \tpfaff_I(z).
\]
This is because for odd $|I|$, both sides are zero.  For $|I|=2p$, $p=1, \dots, \lfloor \frac{n}{2} \rfloor$,
$$
\tpfaff_I(\tilde{z}) = (-1)^{i_1 + i_2 + 1} \cdots (-1)^{i_{2p-1} + i_{2p} + 1}\tpfaff_{I}(z)=\sgn(I)\tpfaff_{I}(z).
$$
Thus we have the following Theorem.

\begin{theorem} \label{thm:tildesumPfaff}
Let $z,y$ be skew-symmetric $n \times n$ matrices. Then
$$\langle \sPf (z), \sPf\dual(y) \rangle = \tpfaff(\tilde{z}+y).
$$
\end{theorem}

In Section \ref{signfixsect} we show that if $\G_P$ is planar there is an easily computable valid
edge ordering $\bar{E}$. Our result may be summarized as follows:

\begin{theorem} Let $P$ be a   problem admitting a matchgate formulation $\Gamma=(V,U,E)$ (e.g., a satisfiability problem
as in the first paragraph) such that
\begin{enumerate}
\item There exists a change of basis in $\BC^2$ such that all the $G_i,R_s$ satisfy the Grassmann-Pl\"ucker identities (i.e., all $G_i$ and $R_s$ are simultaneously realizeable) with complementary indexing.
\item There exists a valid edge order (e.g. if $\Gamma$ is any planar bipartite graph)
\end{enumerate}
Normalize $\pi(G)$ (resp. $\t(R)$) so that the first (resp. last) entry is one, say we need to divide by $\a,\b$ respectively (i.e.\ $\alpha = \prod_i \alpha_i$ where $G_i = \alpha_i \sPf(x_i)$ and similarly for $\beta$).
Consider skew symmetric matrices $x,y$  where $x^i_j$ is the entry of (the normalized) $ \pi(G)$ corresponding
to $I=(i,j)$ and $y^i_j$ is the entry of  (the normalized) $ \t(R)$ corresponding to   $I^c=(i,j)$.
Then  the number of satisfying assignments to $P$ is given
by $\a\b\tpfaff(\tilde x+y)$.
\end{theorem}

\begin{example} \label{ex:first}
Figure \ref{fig:DiaDia} shows an example of
\#Pl-Mon-3-NAE-SAT, with an edge order given by a path through the
graph.
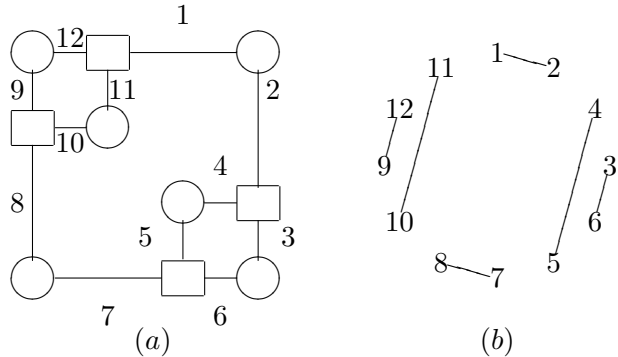
\begin{figure}[htb]
\[
\begin{array}{ccc}
\begin{xy}<10mm,0mm>:
(0,0) *++{\;}*\frm{o};
(0,2) *++{\;}*\frm{-} **@{-};
(0,3) *++{\;}*\frm{o} **@{-};
p+(1,0) *++{\;}*\frm{-} **@{-};
p+(0,-1) *++{\;}*\frm{o} **@{-};
p+(-.7,0) **@{-};
(0.3,0);
(2,0) *++{\;}*\frm{-} **@{-};
(3,0) *++{\;}*\frm{o} **@{-};
p+(0,1) *++{\;}*\frm{-} **@{-};
p+(-1,0) *++{\;}*\frm{o} **@{-};
p+(0,-.75) **@{-};
(3,1.2);
(3,3) *++{\;}*\frm{o} **@{-};
p+(-1.7,0) **@{-};
(-.2,1) *{8};
(-.2,2.5) *{9};
(1.2,2.5) *{11};
(.5,1.8) *{10};
(.5,3.2) *{12};
(3.2,2.5) *{2};
(1,-.5) *{7};
(2.5,-.5) *{6};
(2.5,1.5) *{4};
(1.5,.55) *{5};
(3.4,.55) *{3};
(2,3.5) *{1};
\end{xy}
&\quad &
\begin{xy}<15mm,0mm>:
(0,1);
p+(0,1) *{1}  ="1",
p+(0.5, 0.866) *{2}  ="2",
p+(0.866, 0.5)   *{4} ="3",
p+(1.0, 0)   *{3} ="4",
p+(0.866, -0.5)   *{6} ="5",
p+(0.5, -0.866)   *{5} ="6",
p+(0, -1.0)   *{7} ="7",
p+(-0.5, -0.866)  *{8} ="8",
p+(-0.866, -0.5)  *{10} ="9",
p+(-1.0, 0)  *{9} ="10",
p+(-0.866, 0.5) *{12} ="11",
p+(-0.5, 0.866) *{11} ="12",
"3";"6" **@{-};
"4";"5" **@{-};
"9";"12" **@{-};
"10";"11" **@{-};
"1";"2" **@{-};
"7";"8" **@{-};
\end{xy}
\\
(a) &\quad & (b)
\end{array}
\]
\caption{Example \ref{ex:first}, and the term $S=(1,2)(3,6)(4,5)(7,8)(9,12)(10,11)$
in the Pfaffian (which has no crossings).} \label{fig:DiaDia}
\end{figure}
The corresponding matrix, $\tilde{z}+y$ is below.  In a generator
order, each variable corresponds to a ${ \;\;0\;1 \choose -1\; 0}$
block.  In a recognizer order, each clause corresponds to a $3 \times 3$ block
with $-1/3$ above the diagonal.  Sign flips $z \mapsto \tilde{z}$
occur in a checkerboard pattern with the diagonal flipped; here no
flips occur.  We pick up a factor of $\frac{6}{2^3}$ for each clause
and $2$ for each variable, so $\alpha=2^6$, $\beta=(\frac{6}{2^3})^4$,
and $\alpha \beta \tpfaff(\tilde{z}+y) = 26$ satisfying assignments.
\[
\tilde{z}+y =
{\tiny
\begin{pmatrix}
0  & 1  & 0 &  0 &  0 &  0 &  0 &  0 &  0 & 0  &-1/3&-1/3   \cr
-1 & 0  &-1/3&-1/3& 0 &  0 &  0 & 0 &  0 & 0  & 0  & 0   \cr
 0 &1/3 &  0 &-1/3& 0 & 1  &  0 & 0 & 0 & 0  & 0  & 0   \cr
0  &1/3 &1/3 &  0 & 1 & 0  &  0 &  0 &  0 & 0  & 0  & 0   \cr
0  &  0 &  0 &  -1& 0 &-1/3&-1/3& 0&  0 & 0  & 0  & 0   \cr
0  & -1 &  0 &  0 &1/3&  0 &-1/3&  0 &  0 & 0  & 0  & 0   \cr
0  &  0 &  0 &  0 &1/3&1/3 &  0 & 1 & 0 & 0  & 0  & 0   \cr
0  &  0 &  0 &  0 & 0 &  0 & -1  & 0  &-1/3&-1/3& 0  & 0   \cr
0  &  0 &  0 &  0 & 0 &  0 & 0  &1/3 &  0 &-1/3& 0  & 1   \cr
0  &  0 &  0 &  0 & 0 &  0 &  0 &1/3 &1/3 & 0  & 1 & 0 \cr
1/3  &  0 &  0 &  0 & 0 &  0 &  0 &  0 &  0 & -1  & 0  &-1/3   \cr
1/3& 0  &  0 &  0 & 0 &  0 &  0 &  0 & -1 & 0  &1/3 & 0   \cr
\end{pmatrix}
}
\]
\end{example}

\begin{example} \label{ex:second}
Another  \#Pl-Mon-3-NAE-SAT example which is not read-twice and its $\tilde{z}+y$ matrix are shown in Figure \ref{fig:anoth}.  The
central variable has a submatrix which is again ones above the
diagonal and also contributes $2$ to $\alpha$, so $\alpha=2^5$, $\beta
= (\frac{6}{2^3})^4$.  Four sign changes are necessary in $\tilde{z}$.
 The result is $\alpha \beta \tpfaff(\tilde{z}+y) =14$ satisfying
assignments.
\begin{figure}[htb]
\[
\begin{array}{ccc}
\begin{xy}<7mm,0mm>:
(0,-2) *++{\;}*\frm{o}; p+(1.65,0) **@{-},
p+(0,2) *++{\;}*\frm{-} **@{-}; p+(1.65,0) **@{-},
p+(0,2) *++{\;}*\frm{o} **@{-}; p+(1.65,0) **@{-},
(2,-2) *++{\;}*\frm{-}; p+(1.65,0) **@{-},
p+(0,2) *++{\;}*\frm{o} **@{-}; p+(1.65,0) **@{-},
p+(0,2) *++{\;}*\frm{-} **@{-}; p+(1.65,0) **@{-},
(4,-2) *++{\;}*\frm{o};
p+(0,2) *++{\;}*\frm{-} **@{-};
p+(0,2) *++{\;}*\frm{o} **@{-};
(1,-2.5) *{5};
(3,-2.5) *{6};
(1,-.6) *{1};
(3,-.6) *{3};
(1,2.5) *{10};
(3,2.5) *{9};
(4.4,1) *{8};
(4.4,-1) *{7};
(2.4,1) *{2};
(2.4,-1) *{4};
(-.4,1) *{11};
(-.4,-1) *{12};
\end{xy}
&&
{\tiny
\begin{pmatrix}
0& 1&-1& 1&0&0&0&0&0&0&-1/3&-1/3 \cr
-1&0& 1&-1&0&0&0&0&-1/3&-1/3&0&0 \cr
1&-1&0& 1&0&0&-1/3&-1/3&0&0&0&0 \cr
-1& 1&-1&0&-1/3&-1/3&0&0&0&0&0&0 \cr
0&0&0& 1/3&0&-1/3&0&0&0&0&0& 1 \cr
0&0&0& 1/3& 1/3&0& 1&0&0&0&0&0 \cr
0&0& 1/3&0&0&-1&0&-1/3&0&0&0&0 \cr
0&0& 1/3&0&0&0& 1/3&0& 1&0&0&0 \cr
0& 1/3&0&0&0&0&0&-1&0&-1/3&0&0 \cr
0& 1/3&0&0&0&0&0&0& 1/3&0& 1&0 \cr
1/3&0&0&0&0&0&0&0&0&-1&0&-1/3 \cr
1/3&0&0&0&-1&0&0&0&0&0& 1/3&0 \cr
\end{pmatrix}
}
\end{array}
\]
\caption{Another  \#Pl-Mon-3-NAE-SAT example and its $\tilde{z} + y$ matrix.} \label{fig:anoth}
\end{figure}
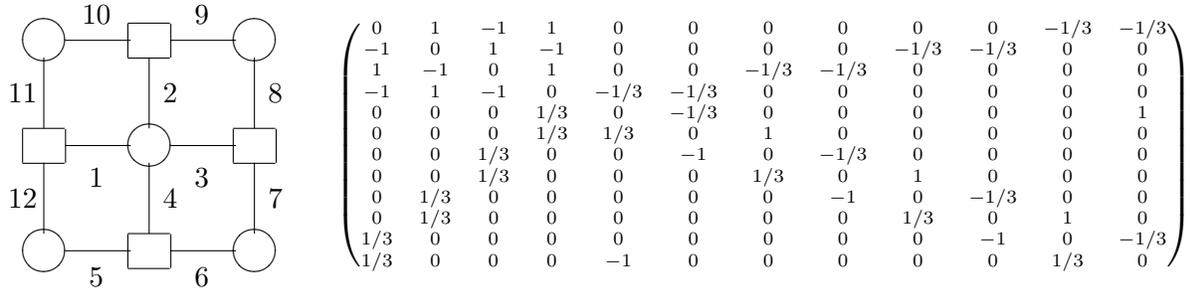

\end{example}

\section{Beyond Pfaffians}\label{CLexample}

As mentioned in the introduction, the key to our approach is that the pairing of a vector in a vector
space of dimension $2^n$ with a vector in its dual space can be accomplished by evaluating a Pfaffian
if both vectors are vectors of Pfaffians of some skew-symmetric matrix. This type of simplification
occurs in many other situations as explained in the Appendix \S\ref{spinsect} below. One simple such
is that if the vector space is of dimension $\binom nk$ and the vectors that are to be paired
are vectors of minors of some $k\times (n-k)$ matrix. Then the pairing can  be done by
computing the determinant of an easily constructed auxiliary $(n-k \times n-k)$ or $(k \times k)$-matrix, so if $k$ is on the order of $\llcorner \frac n2\lrcorner$ there is a spectacular savings.
Explicitly,
 for $k \times \ell$ matrices $z$ and $y$, with $G={\rm sDet}(z)$ and
$R={\rm sDet}(y)$,
$$\langle G,R \rangle = \det(\Id + z^{\top}y).
$$
Here is  an example that exploits this situation.

\begin{example} Given a graph $G$ and an arbitrary orientation of $E(G)$, the \emph{incidence matrix} $B=(b_{v}^{e})_{v \in V(G), \: e \in E(G)}$ is a $|V(G)| \times |E(G)|$ matrix defined by
\begin{equation*}
b_v^e =
\begin{cases} 1 & \text{if $v$ is the initial vertex of $e$,}
\\
-1 &\text{{if $v$ is the terminal vertex of $e$,}}
\\
0 &\text{otherwise.}
\end{cases}
\end{equation*}
For $W \subseteq V(G)$ and $F \subseteq E(G)$, with $|W|=|F|$, let $\Delta_{W,F}(B)$ denote the corresponding minor of $B$. Let ${\rm sDet}(B)=(1,\Delta_{v,\: e}{B},\ldots, \Delta_{W,F}(B),\ldots)$ denote the vector of minors of $B$.

A \emph{rooted spanning forest of $G$} is a pair $(H,W)$, where  $W \subseteq V(G)$, $H$ is a spanning acyclic subgraph of $G$, and  every component of $H$ contains exactly one vertex of $W$. The minor $\Delta_{W,F}(B)$ equals to $\pm 1$ if  $(G|_F,V(G)  - W)$ is a rooted spanning forest, and $\Delta_{W,F}(B)=0$, otherwise. (See~\cite{MR1401006} for a proof of a generalization of this statement to weighted graphs.) Therefore, the value of the pairing $$\langle {\rm sDet}(B^t), {\rm sDet}(B) \rangle = \sum_{W \subseteq V(G)}\sum_{F \subseteq V(G)}(\Delta_{W,F}(B))^2$$ is equal to the number of rooted spanning forests of $G$.  It is shown~\cite{MR1401006} that this value can be computed efficiently by the Cauchy-Binet formula: $$\sum_{W \subseteq V(G)}\sum_{F \subseteq V(G)}(\Delta_{W,F}(B))^2 = \det(\Id+B^tB),$$
where $\Id$ is a $|E(G)|\times|E(G)|$ identity matrix.

From our point of view the result outlined in this example is an instance of the above fact that the pairing of vectors in the Grassmannian and its dual can be computed efficiently. (The Grassmannian can be locally parametrized by vectors of minors of matrices.) The above efficient algorithm for enumerating rooted spanning forests is surprising in the same sense as many holographic algorithms are: A closely related problem of enumerating spanning forests of a graph is $\# P$-hard~\cite{MR1049758}.
\end{example}
\smallskip

\section{Edge ordering and sign}\label{signfixsect}

Throughout this section we assume the local problem has been solved and we only need a valid edge order.  We do not require symmetric signatures.

Given an order $\bar{E}$ we would like to know if it is valid. Say $\bar{E}_G,\bar{E}_R$ are generator and recognizer orders so that
there exist skew-symmetric matrices $z,y$ such that with respect to these orders $G=\sPf(z)$, $R=\sPf(y)$.
Let $\pi,\t\in \FS_{|E|}$ respectively be the permutations such that
 $\pi(\bar{E}_G) = \bar{E}$ and $\t(\bar{E}_R) = \bar{E}$.
Then for all $J\subset [n]$, $\tpfaff_{\pi(J)}(\pi(z))= \sgn(\pi|_J) \tpfaff_J(z)$ 
and similarly for $\t$, so up to signs we have what we want.  Valid orderings yield $\pi, \tau$ which preserve sub-Pfaffian signs.

We now describe one type of valid ordering  for planar graphs, called a {\it $C$-ordering}.
For any planar bipartite graph $\Gamma_P$, a plane curve $C$ intersecting
every edge once corresponds to a non-self-intersecting Eulerian cycle
in the dual of $\Gamma_P$ and can be computed in $O(|E|)$ time.
Fix such a $C$, an orientation and a starting point for $C$, and let $\bar{E}^C$ be the order in which the resulting path crosses the edges of $\Gamma_P$.  Define    $ \bar{E}_G^C$ to be the generator order chosen so that the permutation $\pi : \bar{E}_G^C \ra \bar{E}^C$ is lexicographically minimal.
In particular, $\bar{E}_G^C$ agrees with $\bar{E}_G$ on the edges incident to any fixed generator in $V$. For example, the generator order
on Figure  \ref{fig:DiaDia} is $1,2,3,6,4,5,7,8,9,12,10,11$.
Define  $\bar{E}_R^C$
similarly.

To show that $\bar{E}^C$ is valid we will need another characterization of the sub-Pfaffians and the notion of crossing number.
Let $S=\{(e_1,e_1'),\ldots,(e_k,e_k')\}$ be a partition of an ordered set $I$, with $|I|=2k$, into unordered pairs. Assume, for convenience, that
$e_r < e_r'$ for $1 \leq r \leq k$. Define \emph{the crossing number} $\tcr(S)$ of $S$ as
$$\tcr(S) = \# \{(r,s) \:|\: e_r < e_s < e_r' < e_s' \}.$$
Note that $\tcr(S)$ can be interpreted geometrically as follows. If the elements of $I$ are arranged on a circle in order and the pairs of elements corresponding to pairs in $S$ are joined by straight-line edges, then $\tcr(S)$ is the number of crossings in the resulting geometric graph (see Figure \ref{fig:DiaDia}(b)). When
the order $\bar{E}$ on $I$ is unclear from context we write $\tcr(S,\bar{E})$, instead of $\tcr(S)$.

For $I \subseteq E(\Gamma)$, denote by $\Gamma_I$ the subgraph of $\Gamma$ induced by $I$. Let $\mathscr{S}(\Gamma_I)$ be the set of pairings $S=\{(e_1,e_1'), \dots, (e_k, e_k')\}$ of $I$ such that edges in each pair share a vertex in the set $V$ of generators.  In other words, $(e_i,e_i') \in S$ implies there exists
$j \in V,s,t \in U$ such that $e_i = (j,s), e_i' = (j,t)$.
In what follows we focus on generators, the corresponding statements for recognizers will be clear.

\begin{proposition} \label{prop:crsubpf}
Let $\Gamma$ be a bipartite graph and let $\bar{E}_G$ be a generator edge order. Assume the hypotheses
of Proposition \ref{prop:locglob} are satisfied with
 $z$   the skew-symmetric $|E| \times |E|$ matrix  such that $\sPf(z)=G$ with  the order $\bar{E}_G$.
   Let $I \subset [n] \isom E$.  Then
\[
G_I=\tpfaff_I(z) = \sum_{S \in \mathscr{S}(\Gamma_I)} (-1)^{\tcr(S)}z_S
\]
where $z_S$ is the product $\prod_{(e_i,e_i')\in S} z_{e_i,e_i'}$. 
\end{proposition}
\begin{proof}
Let $\sigma(S)$ denote the permutation
$$\sigma(S) =
( \begin{array}{ccccccc}
e_1 & e_1' & e_2 & e_2' & \ldots & e_k & e_k' \end{array}).
$$
By \cite[p. 91]{MR2399011} or direct verification, $\sgn(\sigma(S))=(-1)^{\tcr(S)}.$
Therefore, for a skew-symmetric matrix $z$ one has $$\tpfaff_I(z)=\sum_{S \in \mathscr{S}} (-1)^{\tcr(S)}z_S,$$ where $z_S:=z_{e_1e_1'}\ldots z_{e_ke_k'}$ and the sum is taken over the set $\mathscr{S}$ of {\it all} partitions of $I$ into pairs.

We need to show that the terms $z_S, S \in \mathscr{S} \setminus \mathscr{S}(\Gamma_I)$ are zero.  Note that for a nonzero term, there must be an even number of edges in the restriction to each variable.  If $S$ contains a pair with split ends $(x_ic_s, x_kc_{t})$, $ i \neq  k$, then $z_{S}=0$.
\end{proof}

The analogous statement to Proposition \ref{prop:crsubpf} holds for recognizers. We can now prove the following Lemma.

\begin{lemma}\label{serglem}
Let $P$ be a problem as above such that
all the associated $G_i,R_s$ satisfy the Grassmann-Pl\"ucker relations under some
change of basis,  $\Gamma_P$ is planar and let  $\bar{E}^C$ be a $C$-ordering. If $\pi, z$ are defined as above, then $\sPf(\pi(z))=\pi(\sPf(z))$.
\end{lemma}

\begin{proof}
It suffices to show that for any $I \subseteq E(\Gamma)$ and any partition $S \in \mathscr{S}(\Gamma_I)$ of $I$, the signs of the term corresponding to $S$ in $\tpfaff_I(z)$ and $\tpfaff_{\pi(I)}(\pi(z))$ are identical. By Proposition \ref{prop:crsubpf}, this is equivalent to showing that
\begin{equation} \label{eq:signcr}
(-1)^{\tcr (S,\bar{E}^C)} = \prod_{x  \in V} (-1)^{\tcr (S \mid_{x },\bar{E}^C_G)},
\end{equation}
where the left hand side of (\ref{eq:signcr}) is the sign of the term corresponding to $S$ appearing in $\tpfaff_{\pi(I)}(\pi(z))$,
and the right hand side is the sign of the term corresponding to $S$  in $\tpfaff_I(z)$, as $$\tpfaff_I(z) = \prod_{x \in V} \tpfaff_{I|_{x }}(z).$$  Here $S|_{x }$ and $I|_{x }$ denote the restriction to the edges incident to ${x }$ of $S$ and $I$, respectively.

A stronger equality, namely $\tcr(S,\bar{E}^C)=\sum_{x \in V} \tcr(S|_x, \bar{E}^C_G)$, holds. The curve $C$ determining $\bar{E}^C$ separates $V$ from $U$. To exploit the geometric intuition presented above, we  replace each vertex in $x \in V$ by a small circle and join the ends of edges in $I$ on this circle by line segments corresponding to pairs in $S|_x$. The total number of crossings in the resulting graph is $\sum_{x \in V} \tcr(S|_x, \bar{E}^C_G)=\sum_{x \in V}\tcr(S|_x, \bar{E}^C),$ as $\bar{E}^C_G$ and $\bar{E}^C$ coincide on the set of edges incident to a fixed $x \in V$. On the other hand, a pair $\{r,s\}$ is counted in $\tcr(S, ,\bar{E}^C)$, if and only if  the curves  with ends on $C$ corresponding to $e_r \cup e_r'$ and $e_s \cup e_s'$ cross.

\begin{figure}
\centering
\includegraphics[trim = 5mm 10mm 110mm 200mm, clip, width=9cm]{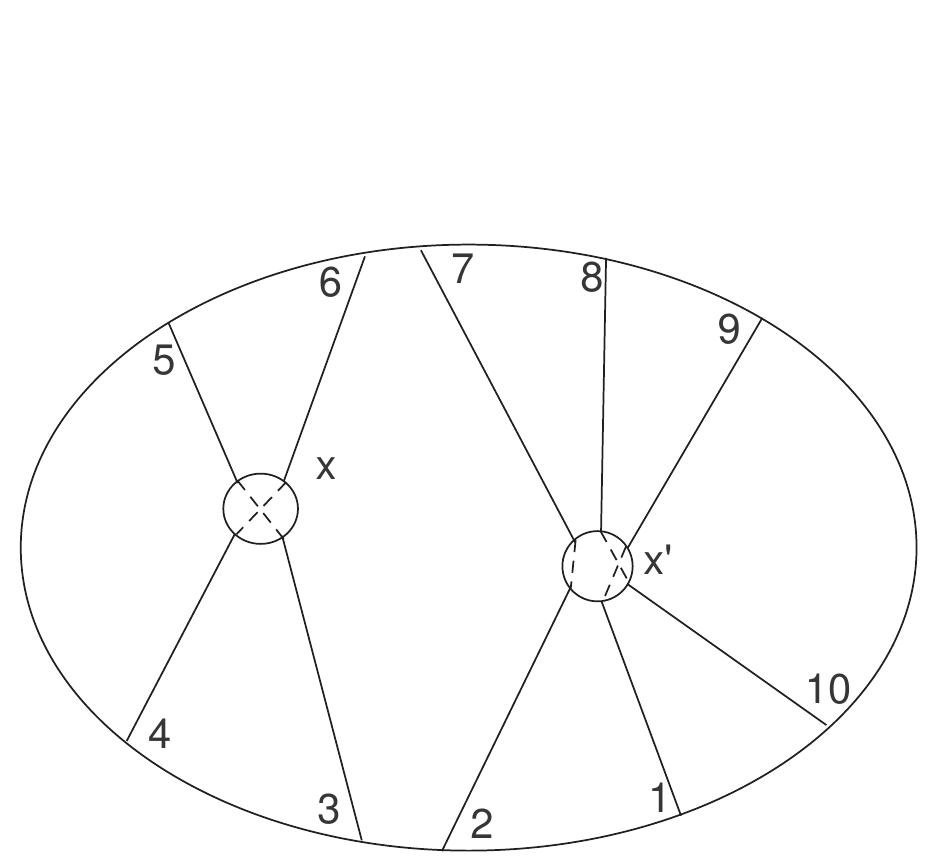} 
\caption{$x,x'$ are two generators, the oval is $C$ and the numbers indicate the ordering
of the edges determined by $C$} \label{f:curve}
\end{figure}

In other words, we are considering restrictions of (the union of $e_r$ and $e_r'$) and (the union $e_s$ and $e_s'$) to the region of the plane bounded by $C$ containing $V$.
\end{proof}
It follows from Lemma \ref{serglem} and a symmetric statement for $\tau$ that $\bar{E}^C$ is valid.
\begin{example}
 An example is given in Figure~\ref{f:curve}. There, the curves composed of edges (3 and 5), and (4 and 6) cross, and that shows that the permutation (3 5 4 6) is odd. The edges corresponding to, say, (3 and 5) and (2 and 7) don't cross, and the permutation (3 5 2 7) is even. In the example, $$S=\{\{1,9\},\{2,7\},\{3,5\},\{4,6\},\{8,10\}\}.$$
The term corresponding to $S$ in $\tpfaff_{\pi(I)}(\pi(x))$ is
$(-1)^{2}x_{1,9}x_{2,7}x_{3,5}x_{4,6}x_{8, 10 }$, as $\tcr(S, \bar{E}^C)=2$.
The term in $\tpfaff(x)$ is a product of $-x_{3,5}x_{4,6}$ and
$-x_{1,9}x_{2,7}x_{8, 10 }$, which are the terms
in Pfaffians of blocks corresponding to $x$ and $x'$, respectively.
\end{example}

\smallskip

\section*{Acknowledgments}
This paper is an outgrowth of the AIM workshop
{\it Geometry and representation theory of tensors for computer science, statistics and other areas}   July 21-25, 2008,
and authors gratefully thank AIM and the other participants of the workshop.
We especially  thank J. Cai, P. Lu,  and  L. Valiant
for their significant efforts to explain their theory to us during the workshop. J. Cai
is also to be thanked
for continuing to answer our questions with extraordinary patience for months afterwards. We thank R. Thomas for
his input on the graph-theoretical part of the argument.
\pagebreak
\bibliographystyle{amsplain}
\bibliography{Lmatrix}

\section{Appendix: Spinors and holographic algorithms}\label{spinsect}

The Grassmann-Pl\"ucker identities are the defining equations for the {\it spinor varieties}
(set of pure spinors).  These equations date back at least to Chevalley in the 1950's \cite{MR1636473}.
The spinor varieties, of which there are two (isomorphic to each other) for each $n$, $\hat \BS_+,\hat\BS_-$, respectively
live in $\La{even}\BC^{n}=:\cS_+$, and
$\La{odd}\BC^{n}=:\cS_-$. The parity condition corresponds to requiring that $G,R$ both
be either in $\cS_+$ or $\cS_-$.  If $n$ is odd then $\cS_+,\cS_-$ are dual vector spaces to
one another, and if $n$ is even, each is self-dual. It is this self-duality that leads
to  the simplification of the exposition with $n$ is even - the discussion for $n$ odd is given below.

 They admit a cover by Zariski open subsets where each subset in e.g. $\hat\BS_+$ is covered
by a map of the form
\begin{align}\label{phimap}
\phi: \La 2\BC^{n}&\ra \oplus_j\La{2j}\BC^{n}=\La{even}\BC^{n}\\
x&\mapsto  (\tpfaff_I(x))
\end{align}
as $I\subseteq (1\hd 2n)$ runs over the subsets of even cardinality (and by convention $\tpfaff_{\emptyset}(x)=1$).

The identification $\cS_+\simeq\La{even}\BC^{2n}$ is not canonical. We can get different identifications by composing
$\phi$ with the action of the Weyl group. The Weyl group action assures that some \lq\lq less convenient\rq\rq\ map will
have first entry nonzero for $G,R$ as mentioned in \S\ref{sect41}.

The map \eqref{phimap} is a special case of a natural map to the \lq\lq big cell\rq\rq\ in
a rational homogeneous variety and the potential generalizations to holographic algorithms
mentioned to in the introduction would correspond to replacing $\hat\BS_+$ by a
Lagrangian Grassmannian or an ordinary Grassmannian of $k$-planes in a $n$-dimensional space.
More generally, if $V$ is a generalized $G(n)$-cominuscule module, where $n$ denotes the
rank of the semi-simple group $G$, then the pairing $V\times V^*\ra \BC$, when restricted
to the cone over the closed orbits in $V,V^*$ can be computed with $O(n^4)$ arithmetic
operations, even though the dimension of $V$ is generally exponential in $n$.

Much of the exposition could be rephrased more concisely using the language of representation theory.
For example, the fact that if each $G_i$ lies in a small spinor variety then
$G=\ot G_i$ lies in a spinor variety as well, is a consequence that the tensor product
of  highest weight vectors   subgroups with compatible Weyl chambers will
be a highest weight vector for the larger group.
Similarly the map $z\mapsto \tilde z$ has a natural interpretation in terms of
an involution on the Clifford module structure that $\cS_+$ comes equipped with.

On the other hand $\BC^{2^n}$ may be viewed as $(\BC^2)^{\ot n}$ and as such, inherits an $SL_2\BC$-action.
The $SL_2(\BC)$ action corresponds to our change of basis, and what we are trying to do is
determine which pairs of points can by simultaneously be moved into the spinor varieties
in $(\BC^2)^{\ot n}$ and the dual space $(\BC^{2*})^{\ot n}$. The convenient basis referred to in the
text corresponds to an identification that embeds the torus of $SL_2$ diagonally into the torus of $Spin_{2n}$
so weight vectors map to weight vectors.

To continue the group perspective in complexity theory more generally,
one can also view the ability to compute the determinant quickly via Gaussian elimination as the consequence
of the robustness of the action of the group preserving the determinant: whereas above there is a subvariety of a huge space (the spinor variety)
on which the pairing can be computed quickly, and a group $SL_2$ that preserves the pairing   - a holographic algorithm can be exploited if the
pair $(G,R)$ can be moved into the subvariety $\hat\BS_+\times \hat\BS_+$ under the action of $SL_2$. In Gaussian elimination, for
the corresponding subvariety one   takes, e.g., the set of upper-triangular matrices, and the group
preserving the determinant acts on the space of matrices sufficiently robustly that any matrix
can be moved into this subvariety (and in polynomial time).
Contrast this with the permanent which is also  easy to evaluate on upper-triangular
matrices, but the group preserving the permanent is not sufficiently robust to
send an arbitrary matrix to an upper-triangular one.  This difference in robustness of group
actions might explain the difference between the determinant and
permanent,  as well as why only solutions to  certain $SAT$ problems can (so far) be counted quickly.

\section{Appendix: Non-symmetric signatures}\label{nonsym}

Most of the natural examples of holographic algorithms, and, in particular, the examples given in this paper, correspond to generator and recognizer signatures $G_i$ and $R_s$ which are \emph{symmetric}, that is invariant under permutations of edges incident to the corresponding vertex. The assumption that the signatures are symmetric is
also convenient for our arguments. If the signatures are symmetric, then the generator tensor $G$ can be represented as a vector of sub-Pfaffians in some generator order if and only if it can be represented as such a vector in every generator order, and the same holds for recognizer orders. This does not hold for general, non-symmetric signatures. We now explain how to deal with non-symmetric signatures.

It is shown in Section~\ref{signfixsect} that given a planar curve $C$, an edge order $\bar{E}^C$ and a generator order $\bar{E}_G^C$, the tensor $G$ can be represented
as a vector of sub-Pfaffians in $\bar{E}^C$ if and only if it can be represented as one in $\bar{E}_G^C$. A similar statement holds for $\bar{E}^C$ and the recognizer order $\bar{E}_R^C$. The edges incident to a given generator are ordered in a clockwise cyclic order in $\bar{E}_G^C$. It is easy to verify that only a cyclic, not linear, ordering enters Grassmann-Pl\"ucker identities. Thus for non-symmetric signatures the following statement holds.

\begin{theorem}
Let $P$ be a   problem admitting a matchgate formulation $\Gamma_P=(V,U,E)$ with $\Gamma_P$ planar. Let the edges incident to every vertex of $\Gamma_P$ be ordered in a clockwise order. Assume that  
there exists a change of basis such that all the $G_i,R_s$ satisfy the Grassmann-Pl\"ucker identities with complementary indexing. Then there exists a valid order, and
the number of satisfying assignments of $P$ can be found in polynomial time.
\end{theorem}

Note that the assumption that the edges (or ``wires'') are ordered in a way that agrees with a planar embedding is also used in the matchgate formulation, as the matchgates must be inserted in such a way that the resulting graph remains planar.

\end{document}